\documentclass[11pt]{article}

\usepackage[a4paper,margin=1in]{geometry}
\usepackage{amsmath,amssymb,amsthm}
\usepackage{mathtools}
\usepackage{mathrsfs}
\usepackage{hyperref}
\usepackage{enumitem}

\hypersetup{
  colorlinks=true,
  linkcolor=blue,
  citecolor=blue,
  urlcolor=blue
}

\newtheorem{theorem}{Theorem}[section]
\newtheorem{proposition}[theorem]{Proposition}

\theoremstyle{definition}
\newtheorem{definition}[theorem]{Definition}
\newtheorem{example}[theorem]{Example}
\theoremstyle{remark}
\newtheorem{remark}[theorem]{Remark}

\newcommand{\grad}{\operatorname{grad}}
\newcommand{\diver}{\operatorname{div}}

\title{Frame Representation of the First-Order Part\\
of the Laplace--Beltrami Operator}
\author{A. G. Nuramatov}
\date{}

\begin{document}

\maketitle

\begin{abstract}

We investigate the geometric content of the first-order part of the
Laplace--Beltrami operator on an oriented Riemannian manifold.
Relative to an arbitrary local orthonormal frame, the first-order part
of the Laplace--Beltrami operator defines a distinguished vector field,
whose coefficients are expressed through the Levi--Civita connection
forms.

The corresponding one-form determines a covariant derivative whose
connection factorization removes the first-order part of the
Laplace--Beltrami operator and replaces it by a natural scalar
potential.

The construction is frame dependent. We show that its transformation
under local rotations of the orthonormal frame compensates the
corresponding variation of the sum-of-squares part, leaving the full
Laplace--Beltrami operator invariant.

These results reveal the geometric role of the first-order part of the
Laplace--Beltrami operator.

\end{abstract}
\section{Introduction}

The Laplace--Beltrami operator is one of the fundamental differential
operators of Riemannian geometry.  Its principal symbol is determined by
the metric and underlies its ellipticity and much of its analytic
behavior.  The lower-order part is equally necessary for invariance, but its 
geometric content is usually hidden within the Christoffel symbols..

The purpose of this paper is to identify the  component of the Levi--Civita
connection that determines  the first-order part of the
Laplace--Beltrami operator. Once this component has been identified, it naturally determines a
covariant derivative whose associated connection factorization removes
the first-order part of the Laplace--Beltrami operator.

In a local orthonormal frame
$\{e_1,\ldots,e_n\}$, the operator takes the form
\[
\Delta=\sum_{i=1}^n e_i^{\,2}+\vec\Omega,
\]
where $\vec\Omega$ is a first-order differential operator determined by
specific contractions of the connection coefficients.  The first-order part of the Laplace--Beltrami operator naturally
determines a distinguished vector field. Its coefficients are expressed
through the Levi--Civita connection forms, while the corresponding
one-form determines a covariant derivative leading to a natural
connection factorization of the Laplace--Beltrami operator.

The construction is local and depends on the chosen orthonormal frame.
This dependence is not a defect: the sum of squares $\sum_i e_i^2$ is
itself frame dependent when the change of frame varies from point to
point.  The variation of $\vec\Omega$ precisely compensates the
additional first-order terms produced by the transformed frame, so that
the complete operator remains invariant.

A particularly transparent interpretation appears when the frame is
induced by an orthogonal coordinate net.  In this case the connection
forms contain no coefficients with three pairwise distinct indices.
Consequently, every nonzero connection coefficient is a component of a
geodesic-curvature vector of a coordinate curve, and the full nonzero
connection data enter the first-order part of the operator.

The paper is organized as follows.  Section~2 fixes the moving-frame and
Hodge-theoretic conventions.  Section~3 constructs a distinguished $(n-1)$-form from the
Levi--Civita connection forms. Its Hodge dual defines the corresponding
one-form and vector field, which are then identified with the
first-order part of the Laplace--Beltrami operator in a moving frame.
Section~4 studies the behavior under local
frame rotations.  Section~5 treats orthogonal coordinate frames.
Section~6 gives a half-connection factorization that removes the
first-order term.  Section~7 describes the three-dimensional Darboux
interpretation.

\section{Moving Frames and Notation}

Let $(M,g)$ be an oriented Riemannian manifold of dimension $n$.  On a
coordinate neighborhood we choose an oriented orthonormal frame
\[
\{e_1,\ldots,e_n\}
\]
and its dual coframe
\[
\{\theta^1,\ldots,\theta^n\},
\qquad
\theta^i(e_j)=\delta^i_j,
\qquad
g=\sum_{i=1}^n \theta^i\otimes\theta^i.
\]

The Levi--Civita connection is represented by the connection one-forms
$\omega^i{}_j$, defined by
\[
\nabla e_j=\omega^i{}_j\,e_i.
\]
We write
\[
\omega^i{}_j=\omega^i{}_{jk}\,\theta^k,
\qquad
\omega^i{}_{jk}=\omega^i{}_j(e_k).
\]
Metric compatibility gives
\[
\omega^i_j+\omega^j_i=0,
\qquad
\omega^i{}_{jk}=-\omega^j{}_{ik}.
\]
The first and second Cartan structure equations are
\[
d\theta^i=-\omega^i{}_j\wedge\theta^j,
\]
\[
\mathcal R^i{}_j
=d\omega^i{}_j+\omega^i{}_k\wedge\omega^k{}_j,
\]
where $\mathcal R^i{}_j$ denotes the curvature two-form. 

For a smooth function $f$, the Hessian is
\[
\nabla^2f(X,Y)=X(Yf)-(\nabla_XY)f,
\]
and the Laplace--Beltrami operator is taken with the sign convention
\[
\Delta f=\operatorname{tr}_g\nabla^2f
=\sum_{i=1}^n\left(e_i e_i f-(\nabla_{e_i}e_i)f\right).
\]
Thus $\Delta=\diver\grad$ in the usual Riemannian convention.

The Hodge operator is denoted by
\[
*:\Lambda^pT^*M\longrightarrow\Lambda^{n-p}T^*M
\]
and is determined by the metric and orientation.  We use
\[
**\alpha=(-1)^{p(n-p)}\alpha,
\qquad \alpha\in\Lambda^pT^*M,
\]
and, for a one-form $\alpha$ and a two-form $\beta$,
\[
*\bigl(\alpha\wedge *\beta\bigr)
=(-1)^n\,\iota_{\alpha^\sharp}\beta.
\]
On one-forms the codifferential satisfies
\[
\delta=-*d*,
\]
although the main factorization formula will be written directly in
terms of $d$ and $*$.

Whenever no confusion can arise, a differential form is also regarded
as the operator of exterior multiplication by that form.  Thus
\[
\Omega\eta=\Omega\wedge\eta
\]
for every differential form $\eta$.
\section{A Distinguished Component of the Levi--Civita Connection}

We now identify the combination of the connection coefficients that
determines the first-order part of the Laplace--Beltrami operator.  The
construction is local and is valid for an arbitrary oriented
orthonormal frame.

\subsection{The $(n-1)$-form associated with the frame}

\begin{definition}
Define the $(n-1)$-form
\[
\widetilde\Omega
=
\sum_{1\le i<j\le n}
\omega^i{}_j\wedge *\bigl(\theta^i\wedge\theta^j\bigr).
\]
\end{definition}

Each connection form $\omega^i{}_j$ represents the infinitesimal
rotation of the frame in the oriented plane spanned by $e_i$ and $e_j$.
The factor $*(\theta^i\wedge\theta^j)$ is the oriented
$(n-2)$-dimensional element orthogonal to that plane.  Thus every
summand combines an infinitesimal rotation with its corresponding
orthogonal moment element.  In this sense $\widetilde\Omega$ may be
viewed as a rotational moment form associated with the chosen frame.

\begin{theorem}\label{thm:hodge-contraction}
The Hodge dual of $\widetilde\Omega$ is the one-form
\[
\Omega
:=(-1)^n*\widetilde\Omega
=
\sum_{i=1}^n
\left(\sum_{j=1}^n\omega^j{}_{ij}\right)\theta^i.
\]
\end{theorem}

\begin{proof}
Applying the Hodge identity
\[
(-1)^n*
\left[
\omega^i{}_j\wedge *\bigl(\theta^i\wedge\theta^j\bigr)
\right]
=
\iota_{(\omega^i{}_j)^\sharp}
\bigl(\theta^i\wedge\theta^j\bigr),
\]
where $(\omega^i{}_j)^\sharp$ denotes the vector field corresponding to the
one-form $\omega^i{}_j$ under the musical isomorphism
$\sharp:T^*M\rightarrow TM$, and $\iota_X$ denotes contraction with the
vector field $X$, we obtain
\[
\iota_{(\omega^i{}_j)^\sharp}
\bigl(\theta^i\wedge\theta^j\bigr)
=
\omega^i{}_{ji}\,\theta^j
-
\omega^i{}_{jj}\,\theta^i,
\]
since
\[
\iota_X(\alpha\wedge\beta)
=
\alpha(X)\beta-\beta(X)\alpha
\]
for arbitrary one-forms $\alpha,\beta$.

Summing over all pairs $i<j$, each coefficient
$\omega^i{}_{ji}$ appears exactly once with the basis one-form
$\theta^j$. Hence
\[
(-1)^n
*
\sum_{i<j}
\left[
\omega^i{}_j
\wedge
*
(\theta^i\wedge\theta^j)
\right]
=
\sum_{j=1}^{n}
\left(
\sum_{i=1}^{n}
\omega^i{}_{ji}
\right)\theta^j.
\]

Applying the musical isomorphism $\sharp$ to the resulting one-form gives
the vector field
\[
\vec{\Omega}
=
\sum_{j=1}^{n}
\left(
\sum_{i=1}^{n}
\omega^i{}_{ji}
\right)e_j,
\]
which is precisely the vector field formed by the coefficients of the
first-order part of the Laplace--Beltrami operator.
\end{proof}

\begin{remark}
No orthogonal-coordinate assumption is used in
Theorem~\ref{thm:hodge-contraction}.  If
\[
\omega^i{}_j=\omega^i{}_{jk}\theta^k,
\]
then coefficients with $i,j,k$ pairwise distinct disappear because
\[
\iota_{e_k}(\theta^i\wedge\theta^j)=0
\qquad (k\ne i,j).
\]
Thus the Hodge contraction selects precisely the components that can
contribute to the first-order term.
\end{remark}

\subsection{The first-order vector field}

Let
\[
\vec\Omega:=\Omega^\sharp
=
\sum_{i=1}^n
\left(\sum_{j=1}^n\omega^j{}_{ij}\right)e_i.
\]
The arrow emphasizes that this vector field acts on functions as a
first-order differential operator.

\begin{theorem}\label{thm:laplace-frame}
In every local oriented orthonormal frame,
\[
\Delta
=
\sum_{i=1}^n e_i^{\,2}+\vec\Omega.
\]
Equivalently, for every $f\in C^\infty(M)$,
\[
\Delta f
=
\sum_{i=1}^n e_i e_i f
+
\sum_{i=1}^n
\left(\sum_{j=1}^n\omega^j{}_{ij}\right)e_i f.
\]
\end{theorem}

\begin{proof}
From the definition of the Laplace--Beltrami operator,
\[
\Delta f
=\sum_{i=1}^n
\left(e_i e_i f-(\nabla_{e_i}e_i)f\right).
\]
Now
\[
\nabla_{e_i}e_i
=\sum_{k=1}^n\omega^k{}_{ii}e_k
=-\sum_{k=1}^n\omega^i{}_{ki}e_k,
\]
where skew-symmetry was used in the last equality.  Relabeling the dummy
indices gives the stated formula.
\end{proof}

\begin{remark}
The connection contains more information than the scalar
Laplace--Beltrami operator uses.  The vectors $\nabla_{e_i}e_i$ describe
the self-acceleration, or geodesic-curvature vectors, of the frame
trajectories.  Coefficients $\omega^i{}_{jk}$ with three pairwise
distinct indices describe mutual rotation of different frame
directions.  They do not enter the first-order part in the chosen
frame, although they remain essential for the full covariant derivative
and curvature.
\end{remark}
\section{Transformation under Local Frame Rotations}

The decomposition in Theorem~\ref{thm:laplace-frame} is associated with
a chosen local orthonormal frame.  Neither the sum of squares nor
$\vec\Omega$ is separately invariant under a point-dependent rotation
of the frame.

Let
\[
e'_a=A_a{}^i e_i,
\qquad A:U\to SO(n),
\]

and let $\theta'^a=A^a{}_i\theta^i$ be the dual coframe.  The connection
matrix transforms by
\[
\omega'=A\omega A^{-1}+A\,dA^{-1}.
\]
The inhomogeneous term is the source of the additional first-order part
in the transformed expression.

\begin{proposition}\label{prop:compensation}
Let $\vec\Omega$ and $\vec\Omega'$ be the vector fields constructed from
the frames $\{e_i\}$ and $\{e'_a\}$, respectively.  Then
\[
\sum_{a=1}^n(e'_a)^2+\vec\Omega'
=
\sum_{i=1}^n e_i^2+\vec\Omega
=
\Delta.
\]
Hence the variations of the two frame-dependent terms compensate
exactly.
\end{proposition}

\begin{proof}
Apply Theorem~\ref{thm:laplace-frame} to each orthonormal frame.  Both
expressions equal the intrinsic Laplace--Beltrami operator.
\end{proof}

\begin{example}[A variable rotation of the Euclidean frame]
On $\mathbb R^2$ with its Euclidean metric, set
\[
e'_1=\cos\varphi\,\partial_x+\sin\varphi\,\partial_y,
\qquad
 e'_2=-\sin\varphi\,\partial_x+\cos\varphi\,\partial_y,
\]
where $\varphi=\varphi(x,y)$.  The dual coframe is
\[
\theta'^1=\cos\varphi\,dx+\sin\varphi\,dy,
\qquad
\theta'^2=-\sin\varphi\,dx+\cos\varphi\,dy.
\]
With the convention $\nabla e'_j=\omega'^i{}_j e'_i$,
\[
\omega'^1{}_2=-d\varphi.
\]
In dimension two, $*(\theta'^1\wedge\theta'^2)=1$, and therefore
\[
\widetilde\Omega'=-d\varphi,
\qquad
\Omega'=-*d\varphi.
\]
Thus
\[
\vec\Omega'
=\varphi_y\,\partial_x-\varphi_x\,\partial_y.
\]
A direct expansion gives
\[
(e'_1)^2+(e'_2)^2
=
\partial_x^2+\partial_y^2
-\varphi_y\partial_x+\varphi_x\partial_y,
\]
and hence
\[
(e'_1)^2+(e'_2)^2+\vec\Omega'
=
\partial_x^2+\partial_y^2.
\]
\end{example}

\begin{remark}
The forms $\widetilde\Omega$ and $\Omega$, and the vector field
$\vec\Omega$, should therefore be understood as frame-associated
quantities.  The intrinsic object is the complete operator
$\sum_i e_i^2+\vec\Omega$.
\end{remark}
\section{Orthogonal Coordinate Frames}

The general construction acquires a particularly transparent geometric
interpretation when the orthonormal frame is induced by an orthogonal
coordinate net.  Let $(x^1,\ldots,x^n)$ be orthogonal coordinates, so
that
\[
g=\sum_{i=1}^n H_i^2 (dx^i)^2,
\qquad
\theta^i=H_i\,dx^i,
\qquad
 e_i=\frac1{H_i}\frac{\partial}{\partial x^i}.
\]

For such a frame the connection forms have the classical special form
\[
\omega^i{}_j
=
\omega^i{}_{ji}\,\theta^i
+
\omega^i{}_{jj}\,\theta^j,
\qquad i\ne j,
\]
or, equivalently,
\[
\omega^i{}_{jk}=0
\qquad\text{whenever }i,j,k\text{ are pairwise distinct}.
\]

For each coordinate trajectory tangent to $e_i$, its
geodesic-curvature vector is
\[
\kappa_i
:=\nabla_{e_i}e_i
=
\sum_{j\ne i}\omega^j{}_{ii}e_j.
\]
Thus every nonzero coefficient of the connection is a component of the
geodesic-curvature vector of one of the coordinate curves.

The distinguished one-form and vector field become
\[
\Omega
=
\sum_{i=1}^n
\left(\sum_{j=1}^n\omega^j{}_{ij}\right)\theta^i,
\]
\[
\vec\Omega
=
\sum_{i=1}^n
\left(\sum_{j=1}^n\omega^j{}_{ij}\right)e_i.
\]
Consequently,
\[
\Delta
=
\sum_{i=1}^n e_i^{\,2}
+
\sum_{i=1}^n
\left(\sum_{j=1}^n\omega^j{}_{ij}\right)e_i.
\]

\begin{proposition}
For an orthogonal coordinate frame, every nonzero coefficient of the
Levi--Civita connection is a component of the geodesic-curvature vector
of a coordinate trajectory.  In particular, there are no additional
coefficients with three pairwise distinct indices that are invisible to
the first-order term.
\end{proposition}

\begin{proof}
The special form of $\omega^i{}_j$ shows that only the coefficients of
$\theta^i$ and $\theta^j$ can occur.  By skew-symmetry these coefficients
are precisely the components appearing in $\nabla_{e_i}e_i$ and
$\nabla_{e_j}e_j$.
\end{proof}

\begin{remark}
This explains the distinguished role of orthogonal coordinate systems.
For a general orthonormal frame the connection also contains mutual
rotations of different frame directions.  In an orthogonal coordinate
frame those coefficients vanish, and the nonzero connection data are
completely encoded by the curvature vectors of the coordinate curves.
\end{remark}

Under separate reparametrizations $x^i\mapsto f_i(x^i)$ preserving the
same oriented coordinate net, the normalized coframe $\theta^i$ and
frame $e_i$ remain unchanged.  Hence the associated quantities
$\widetilde\Omega$, $\Omega$, and $\vec\Omega$ depend on the oriented net
rather than on the particular parameters chosen along its coordinate
curves.
\section{Canonical Covariant Derivative and Connection Factorization}

The elimination of first-order terms by completing the square is a
classical technique in the theory of second-order differential
operators. In our previous work~\cite{Nuramatov}, this factorization was
obtained by a scalar gauge transformation associated with an exact
one-form. We now show that the same construction admits a natural
geometric interpretation in terms of the distinguished one-form
$\Omega$ introduced in the preceding sections.

Define the first-order differential operator

\[
D=d-\frac12\,\Omega\wedge .
\]

This operator may be regarded as the covariant derivative associated
with the connection one-form $-\Omega/2$. Unlike the classical gauge
construction, no assumption is made that $\Omega$ is closed or exact.

\begin{theorem}
Let

\[
\Delta
=
\sum_{i=1}^{n}e_i^{\,2}
+
\vec{\Omega}
\]

be the moving-frame representation of the Laplace--Beltrami operator.

Then

\[
*
D
*
D
=
\Delta
-
\vec{\Omega}
+
\frac12\delta\Omega
+
\frac14|\Omega|^2.
\]

Consequently,

\[
*
D
*
D
=
\sum_{i=1}^{n}e_i^{\,2}
+
\frac12\delta\Omega
+
\frac14|\Omega|^2.
\]

Thus the first-order part of the Laplace--Beltrami operator is removed by
the covariant derivative \(D\).
\end{theorem}

\begin{proof}

For a smooth function $\psi$,

\[
D\psi
=
d\psi
-
\frac12\psi\,\Omega .
\]

Hence

\[
\begin{aligned}
*
D
*
D\psi
&=
*
\left(
d-\frac12\Omega\wedge
\right)
*
\left(
d\psi-\frac12\psi\Omega
\right)
\\
&=
*d*d\psi
-\frac12*d(\psi*\Omega)
-\frac12*(\Omega\wedge*d\psi)
+\frac14\psi*(\Omega\wedge*\Omega).
\end{aligned}
\]

Using the identities established in Section~2 together with

\[
\delta=-*d*,
\]

we obtain

\[
\begin{aligned}
*d(\psi*\Omega)
&=
*(d\psi\wedge*\Omega)
+
\psi\,*d*\Omega
\\
&=
\vec{\Omega}(\psi)
-
(\delta\Omega)\psi ,
\end{aligned}
\]

while

\[
*(\Omega\wedge*d\psi)
=
\vec{\Omega}(\psi),
\]

and

\[
*(\Omega\wedge*\Omega)
=
|\Omega|^2.
\]

Therefore,

\[
*
D
*
D\psi
=
\Delta\psi
-
\vec{\Omega}(\psi)
+
\left(
\frac12\delta\Omega
+
\frac14|\Omega|^2
\right)\psi .
\]

Hence

\[
*
D
*
D
=
\Delta
-
\vec{\Omega}
+
\frac12\delta\Omega
+
\frac14|\Omega|^2 .
\]

Finally, using the decomposition of the Laplace--Beltrami operator
obtained in Section~3,

\[
\Delta
=
\sum_{i=1}^{n}e_i^{\,2}
+
\vec{\Omega},
\]

we obtain

\[
*
D
*
D
=
\sum_{i=1}^{n}e_i^{\,2}
+
\frac12\delta\Omega
+
\frac14|\Omega|^2.
\]

\end{proof}

The above theorem provides a geometric interpretation of the first-order
part of the Laplace--Beltrami operator. The distinguished one-form
$\Omega$ canonically determines the covariant derivative

\[
D=d-\frac12\Omega\wedge,
\]

whose associated second-order operator removes the first-order part of
the Laplace--Beltrami operator and replaces it by the scalar potential

\[
V_\Omega
=
\frac12\delta\Omega
+
\frac14|\Omega|^2.
\]

The coefficient \(1/2\) is uniquely determined by the cancellation of
the first-order terms and is therefore intrinsic to the factorization.

Thus the Laplace--Beltrami operator canonically determines the
distinguished vector field \(\vec{\Omega}\), the associated one-form
\(\Omega\), the covariant derivative

\[
D=d-\frac12\Omega\wedge,
\]

and the scalar potential

\[
V_\Omega
=
\frac12\delta\Omega
+
\frac14|\Omega|^2.
\]

This completes the geometric interpretation of the first-order part of
the Laplace--Beltrami operator developed in the present paper.
\section{Three-Dimensional Interpretation}

Assume now that $n=3$ and that
\[
\theta^1\wedge\theta^2\wedge\theta^3
\]
is positively oriented.  Then
\[
*(\theta^1\wedge\theta^2)=\theta^3,
\qquad
*(\theta^2\wedge\theta^3)=\theta^1,
\qquad
*(\theta^3\wedge\theta^1)=\theta^2.
\]
Therefore
\[
\widetilde\Omega
=
\omega^1{}_2\wedge\theta^3
+
\omega^2{}_3\wedge\theta^1
+
\omega^3{}_1\wedge\theta^2.
\]

Introduce the Darboux rotation one-form with values in the tangent
bundle,
\[
\boldsymbol\omega_D
=
\omega^2{}_3\,e_1
+
\omega^3{}_1\,e_2
+
\omega^1{}_2\,e_3.
\]
It satisfies
\[
d e_i=\boldsymbol\omega_D\times e_i,
\]
where the equation is understood as an equality of tangent-vector-valued
one-forms.

The expression for $\widetilde\Omega$ pairs each infinitesimal rotation
in a coordinate plane with the orthogonal one-dimensional moment
element.  Hence, in three dimensions, $\widetilde\Omega$ is the scalar
form counterpart of the Darboux rotation form.  The one-form
\[
\Omega=-*\widetilde\Omega
\]
and its dual vector field $\vec\Omega$ retain precisely the contraction
of this rotational information that enters the scalar
Laplace--Beltrami operator.

\begin{remark}
The Darboux form contains the full infinitesimal rotation of the frame,
whereas $\vec\Omega$ contains only the contracted component required by
the scalar Laplacian.  For an orthogonal coordinate frame no connection
coefficients with three distinct indices occur, so this contraction is
completely determined by the curvature vectors of the coordinate
trajectories.
\end{remark}
\section{Conclusion}

We have identified, relative to a local orthonormal frame, the
distinguished part of the Levi--Civita connection that forms the
first-order term of the Laplace--Beltrami operator.  Starting from the
connection forms, the $(n-1)$-form $\widetilde\Omega$ produces the
one-form $\Omega$ by Hodge duality, and the dual vector field
$\vec\Omega$ yields the frame representation
\[
\Delta=\sum_i e_i^2+\vec\Omega.
\]

The individual terms in this decomposition depend on the frame, but
their variations under local frame rotations compensate exactly.  For
orthogonal coordinate frames, coefficients with three pairwise
distinct indices vanish, and every nonzero connection coefficient is a
component of the geodesic-curvature vector of a coordinate curve.  This
explains why the first-order structure is especially transparent for
orthogonal coordinate nets.

Finally, shifting the exterior derivative by $\Omega/2$ removes the
first-order term and produces a scalar geometric potential, while in
three dimensions the construction is naturally related to the Darboux
rotation form.  These observations place the first-order part of the
Laplace--Beltrami operator within the moving-frame geometry of the
Levi--Civita connection.

\end{document}